\documentclass[aps,twocolumn,superscriptaddress]{revtex4}
\usepackage{amsmath,amsthm, amssymb}
\usepackage{bm}
\usepackage{graphicx}
\usepackage{times}
\usepackage{epstopdf}
\usepackage{color} 

\def\textbf#1{{\bf #1}}
\def\be{\begin{equation}}
\def\ee{\end{equation}}
\def\ben{\begin{eqnarray}}
\def\een{\end{eqnarray}}
\def\eea{\end{array}}
\def\bea{\begin{array}}

\newcommand{\bei}{\begin{itemize}}
\newcommand{\eei}{\end{itemize}}
\newcommand{\ket}[1]{|#1\rangle}

\usepackage{cancel}
\usepackage{soul}

\theoremstyle{definition}
\newtheorem{theorem}{Result}
\newtheorem{definition}{Theorem}
\newtheorem{lemma}{Lemma}
\def\x{{\bf x}}
\def\a{{\bf a}}
\def\w{{\bf w}}
\def\E{\mathbb{E}}

\begin{document}

\title{Full randomness from arbitrarily deterministic events}

\author{Rodrigo Gallego}
\affiliation{ICFO-Institut de Ciencies Fotoniques, Av. Carl
Friedrich Gauss, 3, 08860 Castelldefels, Barcelona, Spain}
\author{Lluis Masanes}
\affiliation{ICFO-Institut de Ciencies Fotoniques, Av. Carl
Friedrich Gauss, 3, 08860 Castelldefels, Barcelona, Spain}
\author{Gonzalo De La Torre}
\affiliation{ICFO-Institut de Ciencies Fotoniques, Av. Carl
Friedrich Gauss, 3, 08860 Castelldefels, Barcelona, Spain}
\author{Chirag Dhara}
\affiliation{ICFO-Institut de Ciencies Fotoniques, Av. Carl
Friedrich Gauss, 3, 08860 Castelldefels, Barcelona, Spain}
\author{Leandro Aolita}
\affiliation{ICFO-Institut de Ciencies Fotoniques, Av. Carl
Friedrich Gauss, 3, 08860 Castelldefels, Barcelona, Spain}
\author{Antonio Ac\'\i n}
\affiliation{ICFO-Institut de Ciencies Fotoniques, Av. Carl
Friedrich Gauss, 3, 08860 Castelldefels, Barcelona, Spain}
\affiliation{ICREA-Instituci\'o Catalana de Recerca i Estudis
Avan\c cats, Llu\'{i}s Companys 23, 08010 Barcelona, Spain}

\begin{abstract}
Do completely unpredictable events exist in nature? Classical theory, being fully deterministic, completely excludes fundamental randomness. On the contrary, quantum theory allows for randomness within its axiomatic structure. Yet, the fact that a theory makes prediction only in probabilistic terms does not imply the existence of any form of randomness in nature. The question then remains whether one can certify randomness independent of the physical framework used. While standard Bell tests~\cite{bell} approach this question from this perspective, they require prior perfect randomness, which renders the approach circular. Recently, it has been shown that it is possible to certify full randomness using almost perfect random bits~\cite{cr}. Here, we prove that full randomness can indeed be certified using quantum non-locality under the minimal possible assumptions: the existence of a source of arbitrarily weak (but non-zero) randomness and the impossibility of instantaneous signalling. Thus we are left with a strict dichotomic choice: either our world is fully deterministic or there exist in nature events that are fully random. Apart from the foundational implications, our results represent a quantum protocol for full randomness amplification, an information task known to be impossible classically~\cite{sv}. Finally, they open a new path for device-independent protocols under minimal assumptions.
\end{abstract}
\maketitle
Understanding whether nature is deterministically pre-determined
or there are intrinsically random processes  is a fundamental
question that has attracted the interest of multiple thinkers,
ranging from philosophers and mathematicians to physicists or
neuroscientists. Nowadays this question is also important from a
practical perspective, as random bits constitute a valuable
resource for applications such as cryptographic protocols,
gambling, or the numerical simulation of physical and biological
systems.

Classical physics is a deterministic theory. Perfect knowledge of
the positions and velocities of a system of classical particles at
a given time, as well as of their interactions, allows one to
predict their future (and also past) behavior with total
certainty~\cite{laplace}. Thus, any randomness observed in
classical systems is not intrinsic to the theory but just a
manifestation of our imperfect description of the system.

The advent of quantum physics put into question this deterministic
viewpoint, as there exist experimental situations for which
quantum theory gives predictions only in probabilistic terms, even
if one has a perfect description of the preparation and
interactions of the system. A possible solution to this
classically counterintuitive fact was proposed in the early days
of quantum physics: Quantum mechanics had to be
incomplete~\cite{epr}, and there should be a complete theory
capable of providing deterministic predictions for all conceivable
experiments. There would thus be no room for intrinsic randomness,
and any apparent randomness would again be a consequence of our
lack of control over hypothetical ``hidden variables" not
contemplated by the quantum formalism.

Bell's no-go theorem~\cite{bell}, however, implies that
hidden-variable theories are inconsistent with quantum mechanics.
Therefore, none of these could ever render a deterministic
completion to the quantum formalism. More precisely, all
hidden-variable theories compatible with a local causal structure
predict that any correlations among space-like separated events
satisfy a series of inequalities, known as Bell inequalities. Bell
inequalities, in turn, are violated by some correlations among
quantum particles. This form of correlations defines the
phenomenon of quantum non-locality.

Now, it turns out that quantum non-locality does not  necessarily
imply the existence of fully unpredictable processes in nature.
The reasons behind this are subtle. First of all, unpredictable
processes could be certified only if the no-signalling principle
holds. This states that no instantaneous communication is
possible, which imposes in turn a local causal structure on
events, as in Einstein's special relativity. In fact, Bohm's
theory is both deterministic and able to reproduce all quantum
predictions~\cite{bohm}, but it is incompatible with
no-signalling. Thus, we assume throughout the validity of the
no-signalling principle. Yet, even within the no-signalling
framework, it is still not possible to infer the existence of
fully random processes only from the mere observation of non-local
correlations. This is due to the fact that Bell tests require
measurement settings chosen at random, but the actual randomness
in such choices can never be certified. The extremal example is
given when the settings are determined in advance. Then, any Bell
violation can easily be explained in terms of deterministic
models. As a matter of fact, super-deterministic models, which
postulate that all phenomena in the universe, including our own
mental processes, are fully pre-programmed, are by definition
impossible to rule out.

These considerations imply that the strongest result on the
existence of randomness one can hope for using quantum
non-locality is stated by the following possibility: Given a
source that produces an arbitrarily small but non-zero amount of
randomness, can one still certify the existence of completely
random processes? The main result of this work is to provide an
affirmative answer to this question. Our results, then, imply that
the existence of correlations as those predicted by quantum
physics forces us into a dichotomic choice: Either we postulate
super-deterministic models in which all events in nature are fully
pre-determined, or we accept the existence of fully unpredictable
events.

Besides the philosophical and physics-foundational implications,
our results provide a protocol for perfect randomness
amplification using quantum non-locality. Randomness amplification
is an information-theoretic task whose goal is to use an input
source $\mathcal{S}$ of imperfectly random bits to produce perfect
random bits that are arbitrarily uncorrelated from all the events
that may have been a potential cause of them, i.e. arbitrarily
free. In general, $\mathcal{S}$ produces a sequence of bits $x_1,
x_2, \ldots x_j, \ldots$, with $x_j=0$ or 1 for all $j$, see
Fig.~\ref{lightcone}. Each bit $j$ contains some randomness, in
the sense that the probability $P\,(x_j|e)$ that it takes a given
value $x_j$, conditioned on any pre-existing variable $e$, is such
that
\begin{figure}[t]
\begin{center}
  \includegraphics[width=8.5cm,angle=0]{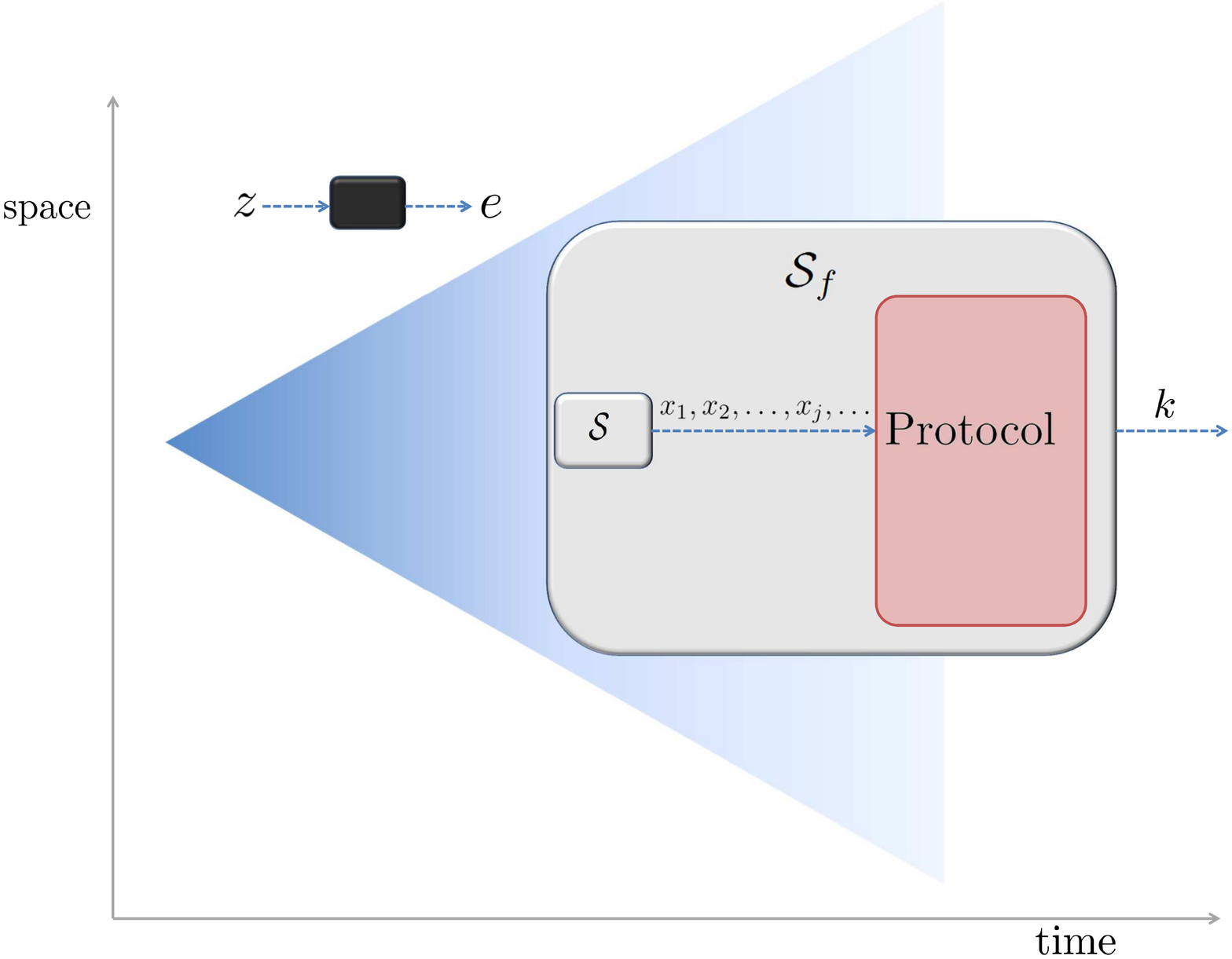}
 \caption{\textbf{Local causal structure and randomness
amplification}. A  source $\mathcal{S}$ produces a sequence $x_1,
x_2, \ldots x_j, \ldots$ Change $x_j$ in the figure to $x_j,
\ldots$ of imperfect random bits. The goal of randomness
amplification is to produce a new source $\mathcal{S}_f$ of
perfect random bits, that is, to process the initial bits to get a
final bit $k$ fully uncorrelated (free) from any potential cause
of it. All space-time events outside the future light-cone of $k$
may have been in its past light-cone before and therefore
constitute a potential cause of it. Any such event can be modeled
by a measurement $z$, with an outcome $e$, on some physical
system. This system may be under the control of an adversary Eve,
interested in predicting the value of $k$.} \label{lightcone}
\end{center}
\end{figure}
\begin{equation}
\label{esource}
\epsilon\leq P\,(x_j|e)\leq 1-\epsilon
\end{equation}
for all $j$ and $e$, where $0<\epsilon\leq 1/2$. The variable $e$
can correspond to any event that could be a possible cause of bit
$x_j$. Therefore, $e$ represents events contained in the
space-time region lying outside the future light-cone of $x_j$.
Free random bits correspond to $\epsilon=\frac{1}{2}$; while
deterministic ones, i.e. those predictable with certainty by an
observer with access to $e$, to $\epsilon=0$. More precisely, when
$\epsilon=0$ the bound \eqref{esource} is trivial and no
randomness can be certified. We refer to $\mathcal{S}$ as an
$\epsilon$-source, and to any bit satisfying \eqref{esource} as an
$\epsilon$-free bit. The aim is then to generate, from arbitrarily
many uses of $\mathcal{S}$, a final source $\mathcal{S}_f$ of
$\epsilon_f$ arbitrarily close to $1/2$. If this is possible, no
cause $e$ can be assigned to the bits produced by $\mathcal{S}_f$,
which are then fully unpredictable. Note that efficiency issues,
such as the rate of uses of $\mathcal{S}$ required per final bit
generated by $\mathcal{S}_f$ do not play any role in randomness
amplification. The relevant figure of merit is just the quality,
measured by $\epsilon_f$, of the final bits. Thus, without loss of
generality, we restrict our analysis to the problem of generating
a single final free random bit $k$.

Santha and Vazirani proved that randomness amplification is
impossible using classical resources~\cite{sv}. This is in a sense
intuitive, in view of the absence of any intrinsic randomness in
classical physics. In the quantum regime, randomness amplification
has been recently studied by Colbeck and Renner~\cite{cr}. There,
$\mathcal{S}$ is used to choose the measurement settings by two
distant observers, Alice and Bob, in a Bell test~\cite{chained}
involving two entangled quantum particles. The measurement outcome
obtained by one of the observers, say Alice, in one of the
experimental runs (also chosen with $\mathcal{S}$) defines the
output random bit. Colbeck and Renner proved how input bits with
very high randomness, of $0.442<\epsilon\leq 0.5$, can be mapped
into arbitrarily free random bits of $\epsilon_f\rightarrow 1/2$,
and conjectured that randomness amplification should be possible
for any initial randomness~\cite{cr}. Our results also solve this
conjecture, as we show that quantum non-locality can be exploited
to attain \emph{full randomness amplification}, i.e. that
$\epsilon_f$ can be made arbitrarily close to $1/2$ for any
$0<\epsilon\leq 1/2$.

Before presenting the ingredients of our proof, it is worth
commenting on previous works on randomness in connection with
quantum non-locality. In~\cite{nature} it was shown how to bound
the intrinsic randomness generated in a Bell test. These bounds
can be used for device-independent randomness expansion, following
a proposal by Colbeck~\cite{colbeck}, and to achieve a quadratic
expansion of the amount of random bits~\cite{nature}
(see~\cite{amp,pm,cwi,vv} for further works on device-independent
randomness expansion). Note however that, in randomness expansion,
one assumes instead, from the very beginning, the existence of an
input seed of free random bits, and the main goal is to expand
this into a larger sequence. The figure of merit there is the
ratio between the length of the final and initial strings of free
random bits. Finally, other recent works have analyzed how a lack
of randomness in the measurement choices affects a Bell
test~\cite{KPB,bg,hall} and the randomness generated in
it~\cite{singapore}. 

Let us now sketch the realization of our final source
$\mathcal{S}_f$. We use the input  $\epsilon$-source $\mathcal{S}$
to choose the measurement settings in a multipartite Bell test
involving a number of observers that depends both on the input
$\epsilon$ and the target $\epsilon_f$. After verifying that the
expected Bell violation is obtained, the measurement outcomes are
combined to define the final bit $k$. For pedagogical reasons, we
adopt a cryptographic perspective and assume the worst-case
scenario where all the devices we use may have been prepared by an
adversary Eve equipped with arbitrary non-signalling resources,
possibly even supra-quantum ones. In the preparation, Eve may have
also had access to $\mathcal{S}$ and correlated the bits it
produces with some physical system at her disposal, represented by
a black box in Fig. \ref{lightcone}. Without loss of generality,
we can assume that Eve can reveal the value of $e$ at any stage of
the protocol by measuring this system. Full randomness
amplification is then equivalent to proving that Eve's correlations with $k$ can be made arbitrarily small.

\begin{figure*}[ht!]
\begin{center}
  \includegraphics[width=15cm,angle=0]{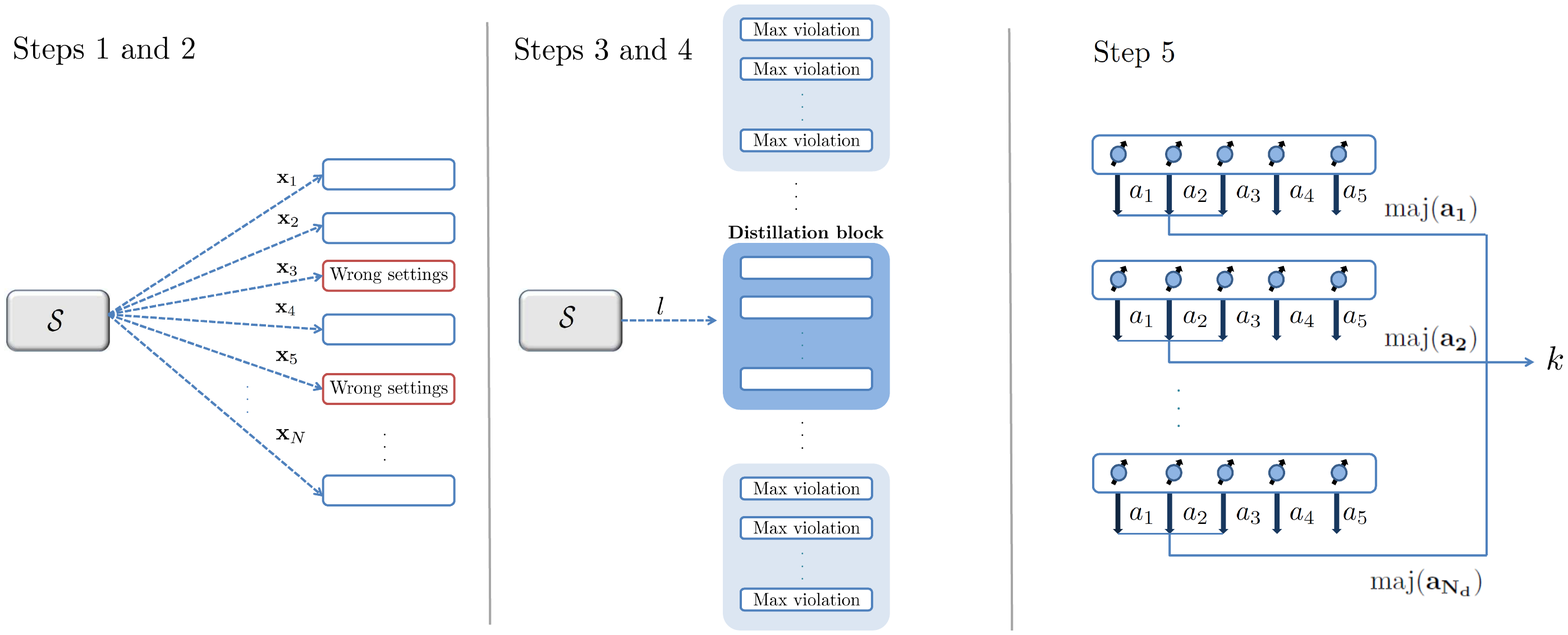}
\caption{\textbf{Protocol for full randomness amplification based
on quantum non-locality}. In the first two steps, all $N$
quintuplets measure their devices, where the choice of measurement
is done using the $\epsilon$-source $\mathcal{S}$; the quintuplets
whose settings happen not to take place in the five-party Mermin
inequality are discarded (in red). In steps 3 and 4, the remaining
quintuplets are grouped into blocks. One of the blocks is chosen
as the distillation block, using again $\mathcal{S}$, while the
others are used to check the Bell violation. In the fifth step,
the random bit $k$ is extracted from the distillation block.}
\label{Fullprotocol}
\end{center}
\end{figure*}

Bell tests for which quantum correlations achieve the maximal
non-signalling violation, also known as
Greenberger-Horne-Zeilinger (GHZ) paradoxes~\cite{ghz}, are
necessary for randomness amplification. This is due to the fact
that unless the maximal non-signalling violation is attained, for
sufficiently small $\epsilon$, Eve may fake the observed
correlations with classical deterministic resources. This attack
ceases to be possible when the maximal non-signalling violation is
observed, as Eve is forced to prepare only those non-local
correlations attaining the maximal violation. GHZ paradoxes are
however not sufficient. Consider for instance the GHZ paradox
given by the tripartite Mermin Bell inequality~\cite{mermin}. One
can see that Eve can predict with certainty any function of the
measurement outcomes and still deliver the maximal violation, for
all $0\leq\epsilon\leq 1/2$ (see Appendix \ref{Sec2}).

For more parties though, the latter happens not to hold any
longer. In fact, consider any correlations attaining the maximal
violation of the five-party Mermin inequality. Take the bit
corresponding to the majority-vote function of the outcomes of any
subset of three out of the five observers, say the first three.
This function is equal to zero if at least two of the three bits
are equal to zero, and equal to one otherwise. We show in Appendix \ref{Sec2} that Eve's predictability on this bit is at
most 3/4. This is our first result:

\begin{theorem}
\label{Theorem1maintext} Given an $\epsilon$-source with any
$0<\epsilon\leq 1/2$, and quantum five-party non-local resources,
an intermediate $\epsilon_i$-source of $\epsilon_i=1/4$ can be
obtained.
\end{theorem}

The partial unpredictability in the five-party Mermin Bell test is
the building block of our protocol. To complete it, we must equip
it with two essential components: ({\it i}) an {\it estimation
procedure} that verifies that the untrusted devices do yield the
required Bell violation; and ({\it ii}) a {\it distillation
procedure} that, from sufficiently many $\epsilon_i$-bits
generated in the 5-party Bell experiment, distills a single final
$\epsilon_f$-source of $\epsilon_f \rightarrow 1/2$. To these
ends, we consider a more complex Bell test involving $N$ groups of
five observers (quintuplets) each, as depicted in Fig.
\ref{Fullprotocol}. The steps in the protocol are described in Box
1.

\begin{table}[t]
\begin{center}
\begin{tabular}{|p{8.05cm}|}
\hline \noindent\textbf{Box 1: Protocol for Randomness
Amplification } \noindent\begin{enumerate}
\item Every observer measures his device in one of two settings chosen at random by the input $\epsilon$-source $\mathcal{S}$.
\item Every quintuplet whose settings combination does not appear in the five-party Mermin Bell test is discarded. If the quintuplets left are fewer than $N/3$, abort.
\item Group the  quintuples left into $N_b$ blocks of equal size $N_d$. Choose a {\it distillation block} at random with $\mathcal{S}$.
\item If the outcomes of any quintuplet not in the distillation block are inconsistent with the maximal violation of the five-party Mermin Bell test, abort.
\item Distill the final bit from the distillation block. This is done in the following way. The
majority vote $\text{maj}(\a)$ among for instance the outcomes
$a_1$, $a_2$ and $a_3$ of the first three users is computed for
each quintuplet. Then, a function $f$ maps the resulting $N_d$
bits into the final bit $k$.
\end{enumerate} \\
\hline
\end{tabular}
\end{center}
\end{table}

In the appendices we prove using techniques
from~\cite{masanes} that, if the protocol is not aborted, the
final bit produced by the protocol is indistinguishable from an
ideal random bit uncorrelated to the eavesdropper. Thus, the
output free random bits satisfy universally-composable
security~\cite{Canetti}, the highest standard of cryptographic
security, and could be used as seed for randomness expansion or
any other protocol.

Finally, we must show that quantum resources can indeed
successfully implement our protocol. It is immediate to see that
the qubit measurements $X$ or $Y$ on the quantum state
$\ket{\Psi}=\frac{1}{\sqrt{2}}(\ket{00000}+\ket{11111})$, with
$\ket{0}$ and $\ket{1}$ the eigenstates of the $Z$ qubit basis,
yield correlations that maximally violate the five-partite Mermin
inequality in question. This completes our main result.

\begin{theorem}[{\bf Main Result}]
\label{Theorem1maintext} Given an $\epsilon$-source with any
$0<\epsilon\leq 1/2$,  a perfect free random bit $k$ can be
obtained using quantum non-local correlations.
\end{theorem}

In summary, we have presented a protocol that, using quantum
non-local resources, attains \emph{full randomness amplification}.
This task is impossible classically and was not known to be
possible in the quantum regime. As our goal was to prove full
randomness amplification, our analysis focuses on the noise-free
case. In fact, the noisy case only makes sense if one does not aim
at perfect random bits and bounds the amount of randomness in the
final bit. Then, it should be possible to adapt our protocol in
order to get a bound on the noise it tolerates. Other open
questions that naturally follow from our results consist of
studying randomness amplification against quantum eavesdroppers,
or the search of protocols in the bipartite scenario.

From a more fundamental perspective, our results imply that there
exist experiments whose outcomes are fully unpredictable. The only
two assumptions for this conclusion are the existence of events
with an arbitrarily small but non-zero amount of randomness and
the validity of the no-signalling principle. Dropping the former
implies accepting a super-determinisitc view where no randomness
exist, so that we experience a fully pre-determined reality. This
possibility is uninteresting from a scientific perspective, and
even uncomfortable from a philosophical one. Dropping the latter,
in turn, implies abandoning a local causal structure for events in
space-time. However, this is one of the  most fundamental notions
of special relativity, and without which even the very meaning of
randomness or predictability would be unclear, as these concepts
implicitly rely on the cause-effect principle.

\small{\textbf{Acknowledgements} We acknowledge support from the
ERC Starting Grant PERCENT, the EU Projects Q-Essence and QCS, the
Spanish MICIIN through a Juan de la Cierva grant and projects
FIS2010-14830, Explora-Intrinqra and CHIST-ERA DIQIP, an FI Grant
of the Generalitat de Catalunya, CatalunyaCaixa, and Fundaci\'o
Privada Cellex, Barcelona.}


\appendix

\begin{widetext}\bigskip
\section{Mermin inequalities}
\label{Sec1}

The 5-party Mermin inequality~\cite{Mermin} plays a central role
in our construction. In each run of this Bell test, measurements
(inputs) $\x= (x_1, \ldots,x_5)$ on five distant black boxes
generate 5 outcomes (outputs) $\a= (a_1, \ldots,a_5)$, distributed
according to a non-signaling conditional probability distribution
$P(\a|\x)$. Both inputs and outputs are bits, as they can take two
possible values, $x_i,a_i\in\{0,1\}$ with $i=1,\ldots,5$. The
inequality can be written as
\begin{equation}
\label{ghz-I}
    \sum_{\a,\x} I (\a,\x) P (\a|\x) \geq 6\ ,
\end{equation}
with coefficients
\begin{equation}\label{I}
    I(\a, \x)= (a_1 \oplus a_2 \oplus a_3 \oplus a_4 \oplus a_5)\, \delta_{\x \in {\cal X}_0} +(a_1 \oplus a_2 \oplus a_3 \oplus a_4 \oplus a_5 \oplus 1)\, \delta_{\x \in{\cal X}_1}\ ,
\end{equation}
where
\begin{equation*}
    \delta_{\x \in {\cal X}_0} = \left\{
        \begin{array}{ll}
            1 & \mbox{ if } \x \in {\cal X}_0 \\
            0 & \mbox{ if } \x \notin {\cal X}_0
        \end{array}
    \right.\ ,
\end{equation*}
and
\begin{align*}
    {\cal X}_0 &= \{ (10000), (01000), (00100), (00010), (00001), (11111) \}    ,
\\
    {\cal X}_1 &= \{ (00111), (01011), (01101), (01110), (10011), (10101), (10110), (11001), (11010), (11100) \}    .
\end{align*}
That is, only half of all possible combinations of inputs, namely
those in ${\cal X} = {\cal X}_0 \cup {\cal X}_1$, appear in the
Bell inequality.

The maximal, non-signalling and algebraic, violation of the
inequality corresponds to the situation in which the left-hand
side of~(\ref{ghz-I}) is zero. The key property of inequality
\eqref{ghz-I} is that its maximal violation can be attained by
quantum correlations. In fact, Mermin inequalities are defined for
an arbitrary number of parties and quantum correlations attain the
maximal non-signalling violation for any odd number of
parties~\cite{merminN}. This violation is always attained by
performing local measurements on a GHZ quantum state.


\section{Partial unpredictability in the five-party Mermin
inequality} \label{Sec2}

Our interest in Mermin inequalities comes from the fact that, for
an odd number of parties, they can be maximally violated by
quantum correlations. These correlations, then, define a GHZ
paradox, which, as explained in the main text, is necessary for
full randomness amplification. As also mentioned in the main text,
GHZ paradoxes are however not sufficient. In fact, it is always
possible to find non-signalling correlations that (i) maximally
violate the 3-party Mermin inequality but (ii) assign a
deterministic value to any function of the measurement outcomes.
This observation can be checked for all unbiased functions mapping
$\{0,1\}^3$ to $\{0,1\}$ (there are  ${8 \choose 4}$ of those)
through a linear program analogous to the one used to prove the
next Theorem. For a larger number of parties, however, some
functions cannot be deterministically fixed to an specific value
while maximally violating a Mermin inequality, as implied by the
following Theorem.

\begin{definition}
\label{Theorem0} Let a five-party non-signaling conditional
probability distribution $P(\a|\x)$ in which inputs $\x= (x_1,
\ldots,x_5)$ and outputs $\a= (a_1, \ldots,a_5)$ are bits.
Consider the bit ${\rm maj} (\a) \in \{0,1\}$ defined by the
majority-vote function of any subset consisting of three of the
five measurement outcomes, say the first three, $a_1$, $a_2$ and
$a_3$. Then, all non-signalling correlations attaining the maximal
violation of the 5-party Mermin inequality are such that the
probability that ${\rm maj} (\a)$ takes a given value, say 0, is
bounded by
\begin{equation}\label{mermin5_unpr}
    1/4\leq P\left({\rm maj} (\a)=0\right)\leq 3/4 .
\end{equation}
\end{definition}

\begin{proof}
This result was obtained by solving a linear program. Therefore,
the proof is numeric, but exact. Formally, let $P(\a|\x)$ be a
$5$-partite no-signaling probability distribution. For $\x=\x_0\in
\cal{X}$, we performed the maximization,
\begin{equation}
\begin{aligned}
 P_{max} = \; &\max_{P}\; P( {\rm maj}(\a)=0|\x_0) \\
&\text{subject to}\\
&  I(\a,\x)\cdot P(\a|\x)=0\\
\end{aligned}
\end{equation}
which yields the value $P_{max}=3/4$. Since the same result holds
for $P( {\rm maj}(\a)=1|\x_0)$, we get the bound $1/4\leq P( {\rm
maj}(\a)=0)\leq 3/4$.

As a further remark, note that a lower bound to $P_{max}$ can
easily be obtained by noticing that one can construct conditional
probability distributions $P(\a|\x)$ that maximally violate
$5$-partite Mermin inequality \eqref{ghz-I} for which at most one
of the output bits (say $a_1$) is deterministically fixed to
either $0$ or $1$. If the other two output bits ($a_2,a_3)$ were
to be completely random, the majority-vote of the three of them
${\rm maj}(a_1,a_2,a_3)$ could be guessed with a probability of
$3/4$. Our numerical results say that this turns out to be an
optimal strategy.

\end{proof}

Theorem~\ref{Theorem0} implies Result 1 in the main text. Moreover
it constitutes the simplest GHZ paradox in which some randomness
can be certified. This paradox is the building block of our
randomness amplification protocol, presented in the next section.

\section{Protocol for full randomness amplification}
\label{Sec3}

In this section, we describe with more details the protocol
summarized in Box 1 of the main text. The protocol uses as
resources the $\epsilon$-source $\mathcal{S}$ and $5N$ quantum
systems. Recall that the bits produced by the source $\mathcal{S}$
are such that the probability $P\,(x_j|e)$ that bit $j$ takes a
given value $x_j$, conditioned on any pre-existing variable $e$,
is bounded by
\begin{equation}
\label{esource} \epsilon\leq P\,(x_j|e)\leq 1-\epsilon ,
\end{equation}
for all $j$ and $e$, where $0<\epsilon\leq 1/2$. The bound, when
applied to $n$-bit strings produced by the $\epsilon$-source,
implies that
\begin{equation}
\label{esourcen} \epsilon^n\leq P\,(x_1,\ldots,x_n|e)\leq
(1-\epsilon)^n .
\end{equation}
Each of the quantum systems is abstractly modeled by a black box
with binary input $x$ and output $a$. The protocol processes
classically the bits generated by $\mathcal{S}$ and by the quantum
boxes. The result of the protocol is a classical symbol $k$,
associated to an abort/no-abort decision. If the protocol is not
aborted, $k$ encodes the final output bit, with possible values 0
or 1. Whereas when the protocol is aborted, no numerical value is
assigned to $k$ but the symbol $\varnothing$ instead, representing
the fact that the bit is empty. The formal steps of the protocol are:

\begin{enumerate}

\item $\mathcal{S}$ is used to generate $N$ quintuple-bits
$\x_{1},\ldots \x_{N}$, which constitute the inputs for the $5N$
boxes. The boxes then provide $N$ output quintuple-bits $\a_1,
\ldots \a_N$.

\item The quintuplets such that $\x \notin {\cal X}$ are
discarded. The protocol is aborted if the number of remaining
quintuplets is less than $N/3$.

\item The quintuplets left after step 2
are organized in $N_b$ blocks each one having $N_d$ quintuplets.
The number $N_b$ of blocks is chosen to be a power of 2. For the
sake of simplicity, we relabel the index running over the
remaining quintuplets, namely $\x_{1},\ldots \x_{N_bN_d}$ and
outputs $\a_{1},\ldots \a_{N_bN_d}$. The input and output of the
$j$-th block are defined as $y_j= (\x_{(j-1)N_d+1}, \ldots
\x_{(j-1)N_d+N_d})$ and $b_j= (\a_{(j-1)N_d+1}, \ldots
\a_{(j-1)N_d+N_d})$ respectively, with $j\in\{1,\ldots,N_b\}$. The
random variable $l\in \{1, \ldots N_b \}$ is generated by using
$\log_2 N_b$ further bits from $\mathcal{S}$. The value of $l$
specifies which block $(b_l, y_l)$ is chosen to generate $k$, i.e.
the distilling block. We define $(\tilde b, \tilde y) = (b_l,
y_l)$. The other $N_b-1$ blocks are used to check the Bell
violation.

\item The function
\begin{equation}\label{r}
    r[b,y] =
    \left\{\begin{array}{ll}
        1 & \mbox{ if }\   I(\a_1,\x_1) = \cdots =I(\a_{N_d}, \x_{N_d}) =0 \\
        0 & \mbox{ otherwise}
    \end{array} \right.\
\end{equation}
tells whether block $(b,y)$ features the right correlations
($r=1$) or the wrong ones ($r= 0$), in the sense of being
compatible with the maximal violation of inequality \eqref{ghz-I}.
This function is computed for all blocks but the distilling one.
The protocols is aborted unless all of them give the right
correlations,
\begin{equation}\label{g}
    g = \prod_{j=1, j\neq l}^{N_b} r[b_j, y_j] =
    \left\{\begin{array}{ll}
        1 & \mbox{ not abort} \\
        0 & \mbox{ abort}
    \end{array} \right.\ .
\end{equation}
Note that the abort/no-abort decision is independent of whether
the distilling block $l$ is right or wrong.

\item If the protocol is not aborted then $k$ is assigned a bit
generated from $b_l= (\a_1, \ldots \a_{N_d})$ as
\begin{equation}\label{k}
    k= f({\rm maj} (\a_1), \ldots {\rm maj} (\a_{N_d}))\ .
\end{equation}
Here $f: \{0,1\}^{N_d} \to \{0,1\}$ is a function characterized in
Lemma~\ref{L1} below, while ${\rm maj} (\a_i) \in \{0,1\}$ is the
majority-vote among the three first bits of the quintuple string
$\a_i$. If the protocol is aborted it sets $k=\varnothing$.
\end{enumerate}

At the end of the protocol, $k$ is potentially correlated with the
settings of the distilling block $\tilde y= y_l$, the bit $g$
in~\eqref{g}, and the bits
\[
    t= [l, (b_1, y_1),\ldots (b_{l-1}, y_{l-1}),
    (b_{l+1}, y_{l+1}), \ldots (b_{N_b}, y_{N_b})] .
\]
Additionally, an eavesdropper Eve might have a physical system
correlated with $k$, which she may measure at any instance of the
protocol. This system is not necessarily classical or quantum, the
only assumption about it is that measuring it does not produce
instantaneous signaling anywhere else. We label all possible
measurements Eve can perform with the classical variable $z$, and
with $e$ the corresponding outcome. In summary, after the
performance of the protocol all the relevant information is $k,
\tilde y, t, g, e, z$, with statistics described by an unknown
conditional probability distribution $P (k, \tilde y, t, g, e|z)$.

To assess the security of our protocol for full randomness
amplification, we have to show that the distribution describing
the protocol when not aborted is indistinguishable from the
distribution $P_{\rm ideal} (k, \tilde y, t, g, e|z g=1) = \frac1
2 P (\tilde y, t, e|z g=1)$ describing an ideal free random bit.
For later purposes, it is convenient to cover the case when the
protocol is aborted with an equivalent notation: if the protocol
is aborted, we define $P (k, \tilde y,t,e|z g=0) =
\delta_k^{\varnothing}\, P (\tilde y,t,e|z g=0)$ and $P_{\rm
ideal} (k, \tilde y,t,e|z g=0) = \delta_k^{\varnothing}\, P
(\tilde y,t,e|z g=0)$, where $\delta_k^{k\rq{}}$ is a
Kronecker\rq{}s delta. In this case, it is immediate that $P=
P_{\rm ideal}$, as the locally generated symbol $\varnothing$ is
always uncorrelated to the environment. To quantify the
indistinguishability between $P$ and $P_{\rm ideal}$, we consider
the scenario in which an observer, having access to all the
information $k,\tilde y, t, g, e, z$, has to correctly distinguish
between these two distributions. We denote by $P({\rm guess})$ the
optimal probability of correctly guessing between the two
distributions. This probability reads
\begin{equation}
\label{Pguess}
    P({\rm guess}) =  \frac 1 2 + \frac 1 4
    \sum_{k, \tilde y, t, g} \max_z \sum_e \Big| P (k, \tilde y,t,g,e|z) - P_{\rm ideal} (k, \tilde y, t, g, e|z) \Big| ,
\end{equation}
where the second term can be understood as (one fourth of) the
variational distance between $P$ and $P_{\rm ideal}$ generalized
to the case when the distributions are conditioned on an input $z$
\cite{PA}. If the protocol is such that this guessing probability
can be made arbitrarily close to 1/2, it generates a distribution
$P$ that is basically undistinguishable from the ideal one. This
is known as \lq\lq{}universally-composable security\rq\rq{}, and
accounts for the strongest notion of cryptographic security
(see~\cite{Canetti} and~\cite{PA}). It implies that the protocol
produces a random bit that is secure (free) in any context. In
particular, it remains secure even if the adversary Eve has access
to $\tilde y$, $t$ and $g$.

Our main result, namely the security of our protocol for full
randomness amplification, follows from the following Theorem.
\begin{definition}[{\bf Main Theorem}]
\label{Theorem1} Consider the previous protocol for randomness
amplification and the conditional probability distribution $P (k,
\tilde y,t,g,e|z)$ describing the statistics of the bits $k,
\tilde y, t, g$ generated during its execution and any possible
system with input $z$ and output $e$ correlated to them. The
probability $P({\rm guess})$ of correctly guessing between this
distribution and the ideal distribution $P_{\rm ideal} (k, \tilde
y, t, g, e|z)$ is such that
\begin{equation}
\label{t1}
    P({\rm guess})
\ \leq\
    \frac 1 2 + \frac{3 \sqrt{N_d}}{2}
    \left[ \alpha^{N_d} +
    2\, N_b^{\log_2(1-\epsilon)}
    \left( 32\beta \epsilon^{-5} \right)^{N_d} \right]\ .
\end{equation}
where $\alpha$ and $\beta$ are real numbers such that 
$0< \alpha <1 <\beta$.
\end{definition}
The right-hand side of~(\ref{t1}) can be made arbitrary close to
$1/2$, for instance by setting $N_b = \left( 32\, \beta\,
\epsilon^{-5} \right)^{2N_d/ |\log_2(1-\epsilon)|}$ and increasing
$N_d$ subject to the fulfillment of the condition $N_d N_b \geq
N/3$. [Note that $\log_2(1-\epsilon) <0$.] In the limit $P({\rm
guess}) \to 1/2$, the bit $k$ generated by the protocol is
indistinguishable from an ideal free random bit.

The proof of Theorem \ref{Theorem1} is provided in the next
section. Before moving to it, we would like to comment on the main
intuitions behind our protocol. As mentioned, the protocol builds
on the 5-party Mermin inequality because it is the simplest GHZ
paradox allowing some randomness certification. The estimation
part, given by step 4, is rather standard and inspired by
estimation techniques introduced in~\cite{bhk}, which were also
used in~\cite{cra} in the context of randomness amplification. The
most subtle part is the distillation of the final bit in step 5.
Naively, and leaving aside estimation issues, one could argue that
it is nothing but a classical processing by means of the function
$f$ of the imperfect random bits obtained via the $N_d$
quintuplets. But this seems in contradiction with the result by
Santha and Vazirani proving that it is impossible to extract by
classical means a perfect free random bit from imperfect
ones~\cite{sva}. This intuition is however wrong. The reason is
because in our protocol the randomness of the imperfect bits is
certified by a Bell violation, which is impossible classically.
Indeed, the Bell certification allows applying techniques similar
to those obtained in Ref.~\cite{PA} in the context of privacy
amplification against non-signalling eavesdroppers. There, it was
shown how to amplify the privacy, that is the unpredictability, of
one of the measurement outcomes of bipartite correlations
violating a Bell inequality. The key point is that the
amplification, or distillation, was attained in a {\it
deterministic} manner. That is, contrary to standard approaches,
the privacy amplification process described in~\cite{PA} does not
consume any randomness. Clearly, these deterministic techniques
are extremely convenient for our randomness amplification
scenario. In fact, the distillation part in our protocol can be
seen as the translation of the privacy amplification techniques of
Ref.~\cite{PA} to our more complex scenario, involving now 5-party
non-local correlations and a function of three of the measurement
outcomes.

\section{Proof of Theorem \ref{Theorem1}}
\label{Sec4}

Before entering the details of the proof of Theorem
\ref{Theorem1}, let us introduce a convenient notation. In what
follows, we sometimes treat conditional probability distributions
as vectors. To avoid ambiguities, we explicitly label the vectors
describing probability distributions with the arguments of the
distributions in upper case. Thus, for example, we denote by
$P({\bf A|X})$ the $(2^5\times 2^5)$-dimensional vector with
components $P(\a|\x)$ for all $\a,\x \in \{0,1\}^5$. We also
denote by $I$ the vector with components $I(\a, \x)$ given
in~(\ref{I}). With this notation, inequality~(\ref{ghz-I}) can be
written as the scalar product
\[
    I \cdot P({\bf A|X}) = \sum_{\a,\x} I (\a,\x) P (\a|\x) \geq 6\ .
 \]
Any probability distribution $P(\a|\x)$ satisfies $C \cdot P({\bf
A|X}) = 1$, where $C$ is the vector with components $C(\a, \x)=
2^{-5}$. We also use this scalar-product notation for full blocks,
as in
\[
    I^{\otimes N_d} \cdot P(B | Y)
    =
    \sum_{\a_1, \ldots \a_{N_d}} \sum_{\x_1 , \ldots \x_{N_d}}
    \left[ \prod_{i=1}^{N_d} I(\a_i, \x_i) \right]
    P( \a_1, \ldots \a_{N_d} | \x_1, \ldots \x_{N_d})\ .
\]
Following our upper/lower-case convention, the vector $P(B|Y,e,z)$
has components $P(b|y,e,z)$ for all $b,y$ but fixed $e,z$.

The proof of Theorem \ref{Theorem1} relies on two crucial lemmas,
which are stated and proven in Sections \ref{sublemma1} and
\ref{sublemma2}, respectively. The first lemma bounds the
distinguishability between the distribution distilled from a block
of $N_d$ quintuplets and the ideal free random bit as function of
the Bell violation~\eqref{ghz-I} in each quintuplet. In
particular, it guarantees that, if the correlations of all
quintuplets in a given block violate inequality \eqref{ghz-I}
sufficiently much, the bit distilled from the block will be
indistinguishable from an ideal free random bit. The second lemma
is required to guarantee that, if the statistics observed in all
blocks but the distilling one are consistent with a maximal
violation of inequality \eqref{ghz-I}, the violation of the
distilling block will be arbitrarily large.

\begin{proof}[{\bf Proof of Theorem \ref{Theorem1}}]
We begin with the identity
\begin{equation}\label{decomp Pguess}
    P({\rm guess}) =
    P(g=0) P({\rm guess}| g=0) + P(g=1) P({\rm guess}| g=1)\ .
\end{equation}
As discussed, when the protocol is aborted ($g=0$) the
distribution generated by the protocol and the ideal one are
indistinguishable. In other words,
\begin{equation}\label{guess g=0}
    P({\rm guess}| g=0) = \frac 1 2\ .
\end{equation}
If $P(g=0) =1$ then the protocol is secure, though in a trivial
fashion. Next we address the non-trivial case where $P(g=1) >0$.

From  formula~(\ref{Pguess}), we have
\begin{eqnarray}
\nonumber &&
    P({\rm guess}|g=1)
\\ \nonumber &=&
    \frac 1 2 + \frac 1 4 \sum_{k, \tilde y, t} \max_z \sum_e
    \Big| P (k, \tilde y, t, e|z, g=1) - \frac 1 2 P (\tilde y, t, e|z, g=1) \Big|
\\ \nonumber &=&
    \frac 1 2 + \frac 1 4 \sum_{\tilde y, t} P(\tilde y, t| g=1)
    \sum_k \max_z \sum_e
    \Big| P (k,e|z, \tilde y, t, g=1) - \frac 1 2 P (e|z, \tilde y, t, g=1) \Big|
\\ \nonumber &\leq&
    \frac 1 2 + \frac 1 4 \sum_{\tilde y, t} P(\tilde y, t| g=1)
    \, 6 \sqrt{N_d}\, \left(\alpha C + \beta I \right)^{\otimes N_d}\!
    \cdot P(\tilde B|\tilde Y, t, g=1)
\\ \nonumber &=&
    \frac 1 2 + \frac{3 \sqrt{N_d}}{2}
    \left(\alpha C + \beta I \right)^{\otimes N_d} \! \cdot
    \sum_{\tilde y, t} P(\tilde y, t| g=1)  P(\tilde B|\tilde Y, t, g=1)
\\ \nonumber &=&
    \frac 1 2 + \frac{3 \sqrt{N_d}}{2}
    \left(\alpha C + \beta I \right)^{\otimes N_d} \! \cdot
    \sum_{t} P(t|g=1)  P(\tilde B|\tilde Y, t, g=1)
\\ \nonumber &=&
    \frac 1 2 + \frac{3 \sqrt{N_d}}{2}
    \left(\alpha C + \beta I \right)^{\otimes N_d} \! \cdot
    \sum_{t} P(\tilde B,t |\tilde Y, g=1)
\\ \label{eq10} &=&
    \frac 1 2 + \frac{3 \sqrt{N_d}}{2}
    \left(\alpha C + \beta I \right)^{\otimes N_d} \! \cdot
    P(\tilde B|\tilde Y, g=1)
\end{eqnarray}
where the inequality is due to Lemma~\ref{lem1} in
Section~\ref{sublemma1}, we have used the no-signalling condition
through $P(\tilde y, t|z, g=1)= P(\tilde y, t| g=1)$, in the
second equality, and Bayes rule in the second and sixth
equalities. From \eqref{eq10} and Lemma~\ref{lemmaestimation} in
Section~\ref{sublemma2}, we obtain
\begin{equation}
\label{boundp}
    P({\rm guess}|g=1)
\ \leq\
    \frac 1 2 + \frac{3 \sqrt{N_d}}{2}
    \left[ \alpha^{N_d} +
    \frac {2\, N_b^{\log_2(1-\epsilon)}} {P(g=1) }
    \left( 32 \beta \epsilon^{-5}\right)^{N_d} \right]\ .
\end{equation}
Finally, substituting bound \eqref{boundp} and
equality~(\ref{guess g=0}) into~(\ref{decomp Pguess}), we obtain
\begin{equation}
    P({\rm guess})
\ \leq\
    \frac 1 2 + \frac{3 \sqrt{N_d}}{2}
    \left[ P(g=1)\, \alpha^{N_d} +
    2\, N_b^{\log_2(1-\epsilon)}
    \left( 32 \beta \epsilon^{-5}\right)^{N_d} \right]\ ,
\end{equation}
which, together with $P(g=1) \leq 1$, implies~(\ref{t1}).
\end{proof}
\subsection{Statement and proof of Lemma~\ref{lem1}}
\label{sublemma1}

As mentioned, Lemma~\ref{lem1} provides a bound on the
distinguishability between the probability distribution obtained
after distilling a block of $N_d$ quintuplets and an ideal free
random bit in terms of the Bell violation~\eqref{ghz-I} in each
quintuplet. The proof of Lemma~\ref{lem1}, in turn, requires two
more lemmas, Lemma~\ref{num} and Lemma~\ref{L1}, stated and proven
in Section~\ref{addlemmas}.

\begin{lemma}
\label{lem1} For each integer $N_d \geq 130$ there exists a
function $f: \{0,1\}^{N_d} \to \{0,1\}$ such that, for any given
$(5 N_d+1)$-partite non-signaling distribution $P(\a_1, \ldots
\a_{N_d} ,e|\x_1, \ldots \x_{N_d}, z) =P(b,e|y,z)$, the random
variable $k= f({\rm maj} (\a_1), \ldots {\rm maj} (\a_{N_d}))$
satisfies
\begin{equation}
\label{l2}
    \sum_k \max_z \sum_e \Big| P(k,e|y, z) - \frac1 2 P(e|y,z) \Big|
\ \leq\
    6 \sqrt{N_d} \left(\alpha C + \beta I \right)^{\otimes N_d} \cdot P(B|Y)
\end{equation}
for all inputs $y=(\x_1, \ldots \x_{N_d})\in {\cal X}^{N_d}$, and
where $\alpha$ and $\beta$ are real numbers such that $0< \alpha
<1 <\beta$.
\end{lemma}

\begin{proof}[{\bf Proof of Lemma~\ref{lem1}}]
For any $\x_0 \in {\cal X}$ let $M_w^{\x_0}$ be the vector with
components $M_w^{\x_0} (\a,\x)= \delta_{{\rm maj} (\a)}^w
\delta_{\x}^{\x_0}$. The probability of getting ${\rm maj} ({\bf
a}) =w$ when using $\x_0$ as input can be written as $P(w|\x_0)=
M_w^{\x_0} \cdot P({\bf A|X})$. Note that this probability can
also be written as $P(w|\x_0)= \Gamma_w^{\x_0} \cdot P({\bf
A|X})$, where $\Gamma_w^{\x_0} = M_w^{\x_0} +\Lambda_w^{\x_0}$ and
$\Lambda_w^{\x_0}$ is any vector orthogonal to the no-signaling
subspace, that is, such that $\Lambda_w^{\x_0} \cdot P({\bf A}|
{\bf X}) = 0$ for all no-signaling distribution $P({\bf A}| {\bf
X})$. We can then write the left-hand side of \eqref{l2} as
\begin{eqnarray}
\nonumber &&
    \sum_k \max_z \sum_e \left| P(k,e|y, z) - \frac 1 2 P(e|y, z) \right|
\\
\nonumber &=&
    \sum_k \max_z \sum_e P(e|y, z) \left| \sum_\w \left( \delta_{f(\w)}^k - \frac 1 2 \right)
    P(\w |y, e,z)  \right|\\
\label{eq17}
 &=&\sum_k \max_z \sum_e P(e|z) \left| \sum_\w \left( \delta_{f(\w)}^k - \frac 1 2 \right)
    \left( \bigotimes_{i=1}^{N_d} \Gamma_{w_i}^{\x_i} \right) \cdot P(B|Y, e,z)  \right|,
\end{eqnarray}
where in the last equality we have used no-signaling through
$P(e|y,z)= P(e|z)$ and the fact that the probability of obtaining
the string of majorities ${\bf w}$ when inputting $y=(\x_1, \ldots
\x_{N_d} ) \in {\cal X}^{N_d}$ can be written as
\begin{eqnarray}\label{e5}
    &&P(\w|y)
=
    \left( \bigotimes_{i=1}^{N_d} \Gamma_{w_i}^{\x_i} \right) \cdot P(B|Y).
\end{eqnarray}

In what follows, the absolute value of vectors is understood to be
component-wise. Bound~\eqref{eq17} can be rewritten as
\begin{eqnarray}
\nonumber &&
    \sum_k \max_z \sum_e \left| P(k,e|y, z) - \frac 1 2 P(e|y, z) \right|
\\
\nonumber &\leq&
    \sum_k \max_z \sum_e P(e|z) \left| \sum_\w \left( \delta_{f(\w)}^k - \frac 1 2 \right)
    \bigotimes_{i=1}^{N_d} \Gamma_{w_i}^{\x_i} \right| \cdot P(B|Y, e,z)
\\  \nonumber &=&
    \sum_k \max_z \left| \sum_\w \left( \delta_{f(\w)}^k - \frac 1 2 \right)
    \bigotimes_{i=1}^{N_d} \Gamma_{w_i}^{\x_i} \right| \cdot
    \left( \sum_e P(e|z) P(B|Y, e,z) \right)
\\ \label{ee7} &=&
    \sum_k \left| \sum_\w \left( \delta_{f(\w)}^k - \frac 1 2 \right)
    \bigotimes_{i=1}^{N_d} \Gamma_{w_i}^{\x_i} \right| \cdot P(B|Y) ,
\end{eqnarray}
where the inequality follows from the fact that all the components
of the vector $P(B|Y,e,z)$ are positive and no-signalling has been
used again through $P(B|Y,z)= P(B|Y)$ in the last equality. The
bound applies to any function $f$ and holds for any choice of
vectors $\Lambda_w^{\x_i}$ in $\Gamma_w^{\x_i}$. In what follows,
we compute this bound for a specific choice of these vectors and
function $f$.

Take $\Lambda_w^{\x_i}$ to be equal to the vectors
$\Lambda_w^{\x_0}$ in Lemma~\ref{num}. These vectors then satisfy
the bounds~(\ref{c1}) and~(\ref{l1}) in the same Lemma. Take $f$
to be equal to the function whose existence is proven in
Lemma~\ref{L1}. Note that the conditions needed for this  Lemma to
apply are satisfied because of bound~(\ref{c1}) in
Lemma~\ref{num}, and because the free parameter $N_d \geq 130$
satisfies $ \left( 3 \sqrt{N_d} \right)^{-1/N_d} \geq \gamma =
0.9732$. With this choice of $f$ and $\Lambda_w^{\x_i}$,
bound~(\ref{ee7}) becomes
\begin{eqnarray}
\nonumber &&
    \sum_k \max_z \sum_e \left| P(k,e|y, z) - \frac 1 2 P(e|y, z) \right|
\\ \nonumber &\leq&
    \sum_k 3 \sqrt{N_d} \left(\bigotimes_{i=1}^{N_d} \Omega^{\x_i} \right) \cdot P(B|Y)
\\ &\leq&
    6 \sqrt{N_d} \left(\alpha C + \beta I \right)^{\otimes N_d} \cdot P(B|Y)\ ,
\end{eqnarray}
where we have used $\Omega^{\x_i} = \sqrt{(\Gamma_0^{\x_i})^2 +
(\Gamma_1^{\x_i})^2}$, $\sum_k 3 =6$, bound~(\ref{c1}) in
Lemma~\ref{num} and bound~(\ref{l1}) in Lemma~\ref{L1}.
\end{proof}

\subsection{Statement and proof of Lemma~\ref{lemmaestimation}}
\label{sublemma2}

In this section we prove Lemma~\ref{lemmaestimation}. This Lemma
bounds the Bell violation in the distillation block in terms of
the probability of not aborting the protocol in step 4 and the
number and size of the blocks, $N_{b}$ and $N_d$.

\begin{lemma}
\label{lemmaestimation} Let $P(b_1, \ldots b_{N_b} | y_1, \ldots
y_{N_b})$ be a $(5 N_d N_b)$-partite no-signaling distribution,
$y_1, \ldots y_{N_b}$ and $l$
the variables generated in steps 2 and 3 of the protocol,
respectively, and $\alpha$ and $\beta$ real numbers such that $0<
\alpha <1 <\beta$; then
\begin{equation}\label{omega1}
    \left(\alpha C + \beta I \right)^{\otimes N_d} \cdot
    P(\tilde B | \tilde Y, g=1)
\ \leq\
    \alpha^{N_d} +
    \frac {2\, N_b^{\log_2(1-\epsilon)}} {P(g=1) }
    \left( 32 \beta \epsilon^{-5}  \right)^{N_d}
    \ .
\end{equation}
\end{lemma}

\begin{proof}[{\bf Proof of Lemma~\ref{lemmaestimation}}]
According to definition~(\ref{r}) we have $I(\a_i, \x_i) \leq
\delta_{r[b,y]}^0$ for all values of $b= (\a_1, \ldots \a_{N_d})$
and $y= (\x_1, \ldots \x_{N_d})$. This also implies $I(\a_i, \x_i)
I(\a_j, \x_j) \leq \delta_{r[b,y]}^0$ and so on. Due to the
property $0<\alpha<1<\beta$, one has that $(\alpha\, 2^{-5})^{N_d
-i} \beta^i \leq \beta^{N_d}$ for any $i=1, \ldots N_d$. All this
in turn implies
\begin{eqnarray}
\nonumber &&
    \prod_{i=1}^{N_d} \left[ \alpha\, 2^{-5} + \beta I_i \right]
\\ \nonumber &=&
    \left( \alpha\, 2^{-5} \right)^{N_d} +
    \left( \alpha\, 2^{-5} \right)^{N_d-1} \beta \sum_i  I_i +
    \left( \alpha\, 2^{-5} \right)^{N_d-2} \beta^2
    \sum_{i\neq j} I_i I_j + \cdots
\\ \nonumber &\leq&
    \left( \alpha\, 2^{-5} \right)^{N_d} +
    \beta^{N_d} \left( \sum_i  I_i + \sum_{i\neq j} I_i I_j + \cdots \right)
\\ \nonumber &\leq&
    \left( \alpha\, 2^{-5} \right)^{N_d} +
    \beta^{N_d} \left( \sum_i  \delta_{r[b,y]}^0 + \sum_{i\neq j} \delta_{r[b,y]}^0 + \cdots \right)
\\ \label{sole1} &\leq&
    \left( \alpha\, 2^{-5} \right)^{N_d} +
    \beta^{N_d} \left( 2^{N_d} -1\right) \delta_{r[b,y]}^0
\ \leq\
    \left( \alpha\, 2^{-5} \right)^{N_d} +
    \left(\beta \, 2\right)^{N_d} \delta_{r[b,y]}^0\ ,
\end{eqnarray}
where $I_i= I(\a_i, \x_i)$. This implies that
\begin{eqnarray}
\nonumber &&
    \left(\alpha C + \beta I \right)^{\otimes N_d} \cdot
    P( B | Y, g=1)
\\ \nonumber &=&
    \sum_{\a_1, \ldots \a_{N_d}} \sum_{\x_1 , \ldots \x_{N_d}}
    \prod_{i=1}^{N_d} \left[ \alpha\, 2^{-5} + \beta I(\a_i, \x_i) \right]
    P( \a_1, \ldots \a_{N_d} | \x_1, \ldots \x_{N_d}, g=1)
\\ \nonumber &\leq&
    \sum_{b,y} \left[ \left(\alpha\, 2^{-5} \right)^{N_d} +
    (2\beta)^{N_d} \delta_{r[b, y]}^0 \right] P(b|y, g=1)
\\ \nonumber &=&
    \alpha^{N_d} \sum_y 2^{-5 N_d} + (2\beta)^{N_d}
    \sum_{y} P(r =0|y, g=1)
\\ \nonumber &=&
    \alpha^{N_d} + (2\beta)^{N_d}
    \sum_{y} P(r =0|y, g=1)
\\ \label{e34} &=&
    \alpha^{N_d} + (2\beta)^{N_d}
    \sum_{y} \frac{P(r=0, y| g=1)}{P(y| g=1)}  \ .
\end{eqnarray}
We can now bound $P(y| g=1)$ taking into account that $y$ denotes
a $5N_d$-bit string generated by the $\epsilon$-source
$\mathcal{S}$ that remains after step 2 in the protocol. Note that
only half of the 32 possible 5-bit inputs $\x$ generated by the
source belong to ${\cal X}$ and remain after step 2. Thus,
$P((\x_1,\ldots,\x_{N_d})\in {\cal X}^{N_d}|g=1) \leq 16^{N_d}
(1-\epsilon)^{5 N_d}$, where we used~\eqref{esourcen}. This,
together with $P((\x_1,\ldots,\x_{N_d})|g=1)\geq \epsilon^{5 N_d}$
implies that
\begin{equation} \label{py}
    P(y|g=1)
    \geq \left( \frac{\epsilon^5}{16 (1-\epsilon)^5}
    \right)^{N_d}.
\end{equation}
Substituting this bound in~(\ref{e34}), and summing over $y$,
gives
\begin{equation}\label{e20}
    \left(\alpha C + \beta I \right)^{\otimes N_d} \cdot
    P( B | Y, g=1)
\ \leq\
    \alpha^{N_d} + (2\beta)^{N_d}\left( \frac{16 (1-\epsilon)^5}{\epsilon^5} \right)^{N_d} P(r =0| g=1)  \ .
\end{equation}
In what follows we use the notation
\[
    P(1_1, 0_2, 1_3, 1_4, \ldots) = P(r[b_1,y_1] =1, r[b_2, y_2] =0, r[b_3, y_3] =1, r[b_4, y_4] =1, \ldots)\ .
\]
According to~(\ref{g}), the protocol aborts ($g=0$) if there is at
least a \lq\lq{}not right\rq\rq{} block ($r[b_j, y_j]=0$ for some
$j\neq l$). While abortion also happens if there are more than one
\lq\lq{}not right\rq\rq{} block, in what follows we lower-bound
$P(g=0)$ by the probability that there is only one \lq\lq{}not
right\rq\rq{} block:
\begin{eqnarray}\nonumber
    1 &\geq& P(g=0)
\\ \nonumber &\geq &
    \sum_{l=1}^{N_b} P(l) \sum_{l\rq{}=1,\, l\rq{} \neq l}^{N_b}
    P(1_1,\ldots 1_{l-1}, 1_{l+1},\ldots 1_{l'-1}, 0_{l'}, 1_{l'+1}, \ldots 1_{N_b})
\\ \nonumber &\geq &
    \sum_{l} P(l) \sum_{l\rq{} \neq l}
    P(1_1,\ldots 1_{l-1}, 1_l, 1_{l+1},\ldots 1_{l'-1}, 0_{l'}, 1_{l'+1}, \ldots 1_{N_b})
\\ \nonumber &= &
    \sum_{l\rq{}} \left[ \mbox{$\sum_{l \neq l\rq{}}$} P(l) \right]
    P(1_1,\ldots 1_{l-1}, 1_l, 1_{l+1},\ldots 1_{l'-1}, 0_{l'}, 1_{l'+1}, \ldots 1_{N_b})
\\ \label{e37} &=&
    \sum_{l'} [1-P(l')]\,
    P(1_1,\ldots 1_{l'-1}, 0_{l'}, 1_{l'+1}, \ldots 1_{N_b}),
\end{eqnarray}
where, when performing the sum over $l$, we have used that
$P(1_1,\ldots 1_{l-1}, 1_l, 1_{l+1},\ldots 1_{l'-1}, 0_{l'},
1_{l'+1}, \ldots 1_{N_b})\equiv P(1_1,\ldots 1_{l'-1}, 0_{l'},
1_{l'+1}, \ldots 1_{N_b})$ does not depend on $l$.
Bound~(\ref{esourcen}) implies
\begin{equation}\label{1-p/p}
    \frac{1-P(l)}{P(l)} \ \geq\  \frac{1- (1-\epsilon)^{\log_2 N_b}}{(1-\epsilon)^{\log_2 N_b}} \ =\  N_b^{\log_2 \frac{1}{1-\epsilon}} -1
    \ \geq\  \frac {N_b^{\log_2 \frac{1}{1-\epsilon}}}{2}\ ,
\end{equation}
where the last inequality holds for sufficiently large $N_b$.
Using this and~(\ref{e37}), we obtain
\begin{eqnarray}\nonumber
    1 &\geq &
    \frac 1 2 \sum_{l'} N_b^{\log_2 \frac{1}{1-\epsilon}}\, P(l')\,
    P(1_1,\ldots 1_{l'-1}, 0_{l'}, 1_{l'+1}, \ldots 1_{N_b})
\\ \label{e22}  &\geq &
    \frac 1 2 \, N_b^{\log_2 \frac{1}{1-\epsilon}}\, P(\tilde r= 0, g=1)\ ,
\end{eqnarray}
where $\tilde r= r[b_l, y_l]$. This together with~(\ref{e20}) implies
\begin{eqnarray}\nonumber
    \left(\alpha C + \beta I \right)^{\otimes N_d} \cdot
    P(\tilde B| \tilde Y ,g=1)
&\leq&
    \alpha^{N_d} + (2\beta)^{N_d}
    \left( \frac{16 (1-\epsilon)^5}{\epsilon^5} \right)^{N_d} P(\tilde{r} =0| g=1)
\\ \label{e40} &\leq &
    \alpha^{N_d} + \frac {2} {P(g=1)}
    \left( \frac{32\beta (1-\epsilon)^5}{\epsilon^5} \right)^{N_d}
    N_b^{\log_2 (1-\epsilon)}\ ,
\end{eqnarray}
where, in the second inequality, Bayes rule was again invoked.
Inequality \eqref{e40}, in turn, implies~(\ref{omega1}).
\end{proof}

\subsection{Statement and proof of the additional Lemmas}
\label{addlemmas}

\begin{lemma}\label{num}
For each $\x_0 \in {\cal X}$ there are three vectors
$\Lambda_0^{\x_0}, \Lambda_1^{\x_0}, \Lambda_2^{\x_0}$ orthogonal
to the non-signaling subspace such that for all $w \in \{0,1\}$
and $\a, \x \in \{0,1\}^5$ they satisfy
\begin{equation}\label{c1}
    \sqrt{ \left[ M_0^{\x_0} (\a, \x) + \Lambda_0^{\x_0} (\a, \x) \right]^2 +
    \left[ M_1^{\x_0} (\a, \x) + \Lambda_1^{\x_0} (\a, \x) \right]^2 }
    \leq
    \alpha C (\a, \x) + \beta I(\a, \x) + \Lambda_2^{\x_0} (\a, \x)
\end{equation}
and
\begin{equation}\label{gm}
    \left| M_w^{\x_0} (\a, \x) + \Lambda_w^{\x_0} (\a, \x) \right|
    \leq \gamma \sqrt{
    \left[ M_0^{\x_0} (\a, \x) + \Lambda_0^{\x_0} (\a, \x) \right]^2 +
    \left[ M_1^{\x_0} (\a, \x) + \Lambda_1^{\x_0} (\a, \x) \right]^2
    }
\end{equation}
where $\alpha= 0.8842$, $\beta= 1.260$ and $\gamma= 0.9732$.
\end{lemma}

\begin{proof}[{\bf Proof of Lemma~\ref{num}}]
The proof of this lemma is numeric but rigorous. It is based on
two linear-programming minimization problems, which are carried
for each value of $\x_0 \in {\cal X}$. We have repeated this
process for different values of $\gamma$, finding that $\gamma=
0.9732$ is roughly the smallest value for which the
linear-programs described below are feasible.

The fact that the vectors $\Lambda_0^{\x_0}, \Lambda_1^{\x_0},
\Lambda_2^{\x_0}$ are orthogonal to the non-signaling subspace can
be written as linear equalities
\begin{equation}\label{ns3}
    D \cdot \Lambda_w^{\x_0} = {\bf 0}
\end{equation}
for $w \in \{0,1,2\}$, where ${\bf 0}$ is the zero vector and $D$
is a matrix whose rows constitute a basis of non-signaling
probability distributions. A geometrical interpretation of
constraint~(\ref{c1}) is that the point in the plane with
coordinates $\left[ M_0^{\x_0} (\a, \x) + \Lambda_0^{\x_0} (\a,
\x) , M_1^{\x_0} (\a, \x) + \Lambda_1^{\x_0} (\a, \x) \right] \in
\mathbb{R}^2$ is inside a circle of radius $\alpha C (\a, \x) +
\beta I(\a, \x) + \Lambda_2^{\x_0} (\a, \x)$ centered at the
origin. All points inside an octagon inscribed in this circle also
satisfy constraint~(\ref{c1}). The points of such an inscribed
octagon are the ones satisfying the following set of linear
constraints:
\begin{eqnarray}\label{c2}
\nonumber &&
    \left[ M_0^{\x_0} (\a, \x) + \Lambda_0^{\x_0} (\a, \x) \right] \eta \cos \theta +
    \left[ M_1^{\x_0} (\a, \x) + \Lambda_1^{\x_0} (\a, \x) \right] \eta \sin \theta
\\ \label{lemma2} &\leq&
    \alpha C (\a, \x) + \beta I(\a, \x) + \Lambda_2^{\x_0} (\a, \x)\ ,
\end{eqnarray}
for all $\theta \in \{\frac{\pi}{8}, \frac{3\pi}{8},
\frac{5\pi}{8}, \frac{7\pi}{8}, \frac{9\pi}{8}, \frac{11\pi}{8},
\frac{13\pi}{8}, \frac{15\pi}{8} \}$, where $\eta =  (\cos
\frac{\pi}{8})^{-1} \approx 1.082$. In other words, the eight
conditions~(\ref{c2}) imply constraint~(\ref{c1}). From now on, we
only consider these eight linear constraints~(\ref{c2}). With a
bit of algebra, one can see that inequality~(\ref{gm}) is
equivalent to the two almost linear inequalities there was an
error in the following equation, as the pre-factor in terms of
$\gamma$ was wrong. Please check what was computed and how it
affects to $\gamma$ and, then, to the value of $N_d$
\begin{equation}\label{ee13}
    \pm \left[ M_w^{\x_0} (\a, \x) + \Lambda_w^{\x_0} (\a, \x) \right]
\ \leq\ \sqrt{\frac{\gamma^2}{1-\gamma^2}}
    \left| M_{\bar w}^{\x_0} (\a, \x) + \Lambda_{\bar w}^{\x_0} (\a, \x) \right|\ ,
\end{equation}
for all $w\in \{0,1\}$, where $\bar w = 1-w$. Clearly, the problem
is not linear because of the absolute values. The computation
described in what follows constitutes a trick to make a good guess
for the signs of the terms in the absolute value of~(\ref{ee13}),
so that the problem can be made linear by adding extra
constraints.

The first computational step consists of a linear-programming
minimization of $\alpha$ subject to the constraints~(\ref{ns3}),
(\ref{c2}), where the minimization is performed over the variables
$\alpha, \beta, \Lambda_0^{\x_0}, \Lambda_1^{\x_0},
\Lambda_2^{\x_0}$. This step serves to guess the signs
\begin{equation}\label{signs}
    \sigma_w (\a, \x) \ =\  \mathrm{sign}\! \left[
    M_w^{\x_0} (\a, \x) + \Lambda_w^{\x_0} (\a, \x) \right]\ ,
\end{equation}
for all $w,\a, \x$, where the value of $\Lambda_w^{\x_0} (\a, \x)$
corresponds to the solution of the above minimization. Once we
have identified all these signs, we can write the
inequalities~(\ref{ee13}) in a linear fashion:
\begin{eqnarray}\label{ee132}
    \sigma_w (\a, \x)
    \left[ M_w^{\x_0} (\a, \x) + \Lambda_w^{\x_0} (\a, \x) \right]
&\geq&  0 \ , \\ \label{ee162}
    \sigma_w (\a, \x)
    \left[ M_w^{\x_0} (\a, \x) + \Lambda_w^{\x_0} (\a, \x) \right]
&\leq& \sqrt{\frac{\gamma^2}{1-\gamma^2}}\,
    \sigma_{\bar w} (\a,\x)
    \left[ M_{\bar w}^{\x_0} (\a, \x) + \Lambda_{\bar w}^{\x_0} (\a, \x) \right]
    \ ,
\end{eqnarray}
for all $w\in \{0,1\}$.

The second computational step consists of a linear-programming
minimization of $\alpha$ subjected to the constraints~(\ref{ns3}),
(\ref{c2}), (\ref{ee132}), (\ref{ee162}), over the variables
$\alpha, \beta, \Lambda_0^{\x_0}, \Lambda_1^{\x_0},
\Lambda_2^{\x_0}$. Clearly, any solution to this problem is also a
solution to the original formulation of the Lemma. The
minimization was performed for any $\x_0 \in {\cal X}$ and the
values of $\alpha, \beta$ turned out to be independent of $\x_0
\in {\cal X}$. These obtained numerical values are the ones
appearing in the formulation of the Lemma.
\end{proof}

Note that Lemma~\ref{num} allows one to bound the predictability
of ${\rm maj} (\a)$ by a linear function of the 5-party Mermin
violation. This can be seen by computing $\Gamma_w^{\x_0} \cdot
P({\bf A|X})$ and applying the bounds in the Lemma. In principle,
one expects this bound to exist, as the predictability is smaller
than one at the point of maximal violation, as proven in
Theorem~\ref{Theorem0}, and equal to one at the point of no
violation. However, we were unable to find it. This is why we had
to resort to the linear optimization technique given above, which
moreover provides the bounds~\eqref{c1} and~\eqref{gm} necessary
for the security proof.

\begin{lemma}\label{L1}
Let $N_d$ be a positive integer and let $\Gamma_w^i (\a, \x)$ be a
given set of real coefficients such that for all $i \in \{1,
\ldots N_d \}$, $w \in \{0,1\}$ and $\a, \x \in \{0,1\}^5$ they
satisfy
\begin{equation}\label{l12}
    \left| \Gamma_w^i (\a, \x) \right|
    \leq \left( 3 \sqrt{N_d} \right)^{-1/N_d} \Omega_i (\a,\x)\ ,
\end{equation}
where $\Omega_i (\a,\x) = \sqrt{\Gamma_0^i (\a, \x)^2 + \Gamma_1^i
(\a, \x)^2}$. There exists a function $f: \{0,1\}^{N_d} \to
\{0,1\}$ such that for each sequence $(\a_1 ,\x_1), \ldots
(\a_{N_d} ,\x_{N_d})$ we have
\begin{equation}\label{l1}
    \left| \sum_\w \left( \delta_{f(\w)}^k - \frac 1 2 \right)
    \prod_{i=1}^{N_d} \Gamma_{w_i}^i (\a_i, \x_i) \right|
\ \leq\
    3 \sqrt{N_d}\, \prod_{i=1}^{N_d}
    \Omega_i (\a_i, \x_i)\ ,
\end{equation}
where the sum runs over all $\w= (w_1, \ldots w_{N_d}) \in
\{0,1\}^{N_d}$.
\end{lemma}

\begin{proof}[{\bf Proof of Lemma~\eqref{L1}}]
First, note that for a sequence $(\a_1 ,\x_1), \ldots (\a_{N_d}
,\x_{N_d})$ for which there is at least one value of $i\in \{1,
\ldots N_d\}$ satisfying $\Gamma_0^i (\a_i, \x_i) = \Gamma_1^i
(\a_i, \x_i) = 0$, both the left-hand side and the right-hand side
of~(\ref{l1}) are equal to zero, hence, inequality~(\ref{l1}) is
satisfied independently of the function $f$. Therefore, in what
follows, we only consider sequences $(\a_1 ,\x_1), \ldots
(\a_{N_d} ,\x_{N_d})$ for which either $\Gamma_0^i (\a_i, \x_i)
\neq 0$ or $\Gamma_1^i (\a_i, \x_i) \neq 0$, for all $i=1, \ldots
N_d$. Or, equivalently, we consider sequences such that
\begin{equation}\label{const4}
    \prod_{i=1}^{N_d} \Omega_i (\a_i, \x_i) >0\ .
\end{equation}

The existence of the function $f$ satisfying~(\ref{l1}) for all
such sequences is shown with a probabilistic argument. We consider
the situation where $f$ is picked from the set of all functions
mapping $\{0,1\}^{N_d}$ to $\{0,1\}$ with uniform probability, and
upper-bound the probability that the chosen function does not
satisfy the constraint~(\ref{l1}) for all $k$ and all sequences
$(\a_1 ,\x_1), \ldots (\a_{N_d} ,\x_{N_d})$
satisfying~(\ref{const4}). This upper bound is shown to be smaller
than one. Therefore there must exist at least one function
satisfying~(\ref{l1}).

For each $\w \in \{0,1\}^{N_d}$ consider the random variable
$F_\w= ( \delta_{f(\w)}^0 - \frac 1 2 ) \in \{\frac 1 2 , - \frac
1 2\}$, where $f$ is picked from the set of all functions mapping
$\{0,1\}^{N_d} \to \{0,1\}$ with uniform distribution. This is
equivalent to saying that the $2^{N_d}$ random variables
$\{F_\w\}_\w$ are independent and identically distributed
according to $\Pr \{F_\w = \pm \frac 1 2 \} = \frac 1 2$. For ease
of notation, let us fix a sequence $(\a_1 ,\x_1), \ldots (\a_{N_d}
,\x_{N_d})$ satisfying~(\ref{const4}) and use the short-hand
notation  $\Gamma_{w_i}^i = \Gamma_{w_i}^i (\a_i ,\x_i)$.

We proceed using the same ideas as in the derivation of the
exponential Chebyshev\rq{}s Inequality. For any $\mu, \nu \geq 0$,
we have
\begin{eqnarray}
\nonumber &&
    \Pr \left\{ \sum_\w F_\w \prod_{i=1}^{N_d} \Gamma_{w_i}^i
    \geq \mu \right\}
\\ \nonumber &=&
    \Pr \left\{ \nu
    \left( -\mu +\sum_\w F_\w \prod_{i=1}^{N_d} \Gamma_{w_i}^i \right)
    \geq 0 \right\}
\\ \nonumber &=&
    \Pr \left\{ \exp\! \left(
    -\nu\mu+ \nu\sum_\w F_\w \prod_{i=1}^{N_d} \Gamma_{w_i}^i
    \right) \geq 1 \right\}
\\ \label{e14.1} &\leq&
    \E\! \left[ \exp\! \left(
    -\nu\mu+ \nu\sum_\w F_\w \prod_{i=1}^{N_d} \Gamma_{w_i}^i
    \right) \right]
\\ \nonumber &=&
    \E\! \left[ {\rm e}^{-\nu\mu} \prod_\w  \exp\! \left(
    \nu F_\w \prod_{i=1}^{N_d} \Gamma_{w_i}^i
    \right) \right]
\\ \label{e14.2} &=&
    {\rm e}^{-\nu\mu} \prod_\w \E\! \left[ \exp\! \left(
    \nu F_\w \prod_{i=1}^{N_d} \Gamma_{w_i}^i
    \right) \right]
\\ \label{omega} &\leq&
    {\rm e}^{-\nu\mu} \prod_\w \E \left[ 1+
    \nu F_\w \prod_{i=1}^{N_d} \Gamma_{w_i}^i
    + \left( \nu F_\w \prod_{i=1}^{N_d} \Gamma_{w_i}^i
    \right)^{\! 2} \right]\ .
\end{eqnarray}
Here $\E$ stands for the average over all $F_\w$. In~(\ref{e14.1})
we have used that any positive random variable $X$ satisfies
$\Pr\{X\geq 1\} \leq \E[X]$. In~(\ref{e14.2}) we have used that
the $\{ F_\w \}_\w$ are independent. Finally, in~(\ref{omega}) we
have used that ${\rm e}^{\eta} \leq 1 +\eta + \eta^2$, which is
only valid if $\eta \leq 1$. Therefore, we must show that
\begin{equation}\label{tocheck}
    \left|\frac \nu 2 \prod_{i=1}^{N_d} \Gamma_{w_i}^i\right| \leq 1,
\end{equation}
which is done below, when setting the value of $\nu$. In what
follows we use the chain of inequalities~(\ref{omega}), the fact
that $\E [F_\w] = 0$ and $\E [F_\w^2] = 1/4$, bound $1+ \eta \leq
{\rm e}^{\eta}$ for $\eta \geq 0$, and the definition $\Omega_i^2
= (\Gamma_0^i)^2 + (\Gamma_1^i)^2$:
\begin{eqnarray}
&& \nonumber
    \Pr \left\{ \sum_\w F_\w \prod_{i=1}^{N_d} \Gamma_{w_i}^i
    \geq \mu \right\}
\\ \nonumber &\leq&
    {\rm e}^{-\nu\mu} \prod_\w \left( 1+
    \E [F_\w]\, \nu \prod_{i=1}^{N_d} \Gamma_{w_i}^i
    + \E [ F_\w^2 ]\, \nu^2 \prod_{i=1}^{N_d} \left(\Gamma_{w_i}^i
    \right)^{2} \right)
\\ \nonumber &=&
    {\rm e}^{-\nu\mu} \prod_\w \left( 1+
    \frac{\nu^2}{4} \prod_{i=1}^{N_d} \left( \Gamma_{w_i}^i
    \right)^2 \right)
\\ \nonumber &\leq&
    {\rm e}^{-\nu\mu} \prod_\w \exp\! \left(
    \frac{\nu^2}{4} \prod_{i=1}^{N_d} \left( \Gamma_{w_i}^i
    \right)^2 \right)
\\ \nonumber &=&
    \exp\! \left( -\nu\mu +\sum_\w
    \frac{\nu^2}{4} \prod_{i=1}^{N_d} \left( \Gamma_{w_i}^i
    \right)^2 \right)
\\ \label{e18} &=&
    \exp\! \left( -\nu\mu +
    \frac{\nu^2}{4} \prod_{i=1}^{N_d} \Omega_i^2 \right)
\end{eqnarray}
In order to optimize this upper bound, we minimize the  exponent
over $\nu$. This is done by differentiating with respect to $\nu$
and equating to zero, which gives
\begin{equation}\label{nu*}
    \nu = 2\, \mu \prod_{i=1}^{N_d} \Omega_i^{-2}\ .
\end{equation}
Note that constraint~(\ref{const4}) implies that the inverse of
$\Omega_i$ exists. Since we assume $\mu \geq 0$, the initial
assumption $\nu \geq 0$ is satisfied by the solution~(\ref{nu*}).
By substituting~(\ref{nu*}) in~(\ref{e18}) and rescaling the free
parameter $\mu$ as
\begin{equation}\label{tm}
    \tilde\mu = \frac{\mu} {\prod_{i=1}^{N_d} \Omega_i} \ ,
\end{equation}
we obtain
\begin{equation}\label{xx}
    \Pr \left\{ \sum_\w F_\w \prod_{i=1}^{N_d} \Gamma_{w_i}^i
    \geq \tilde\mu \prod_{i=1}^{N_d} \Omega_i \right\}
\leq
    {\rm e}^{-\tilde{\mu}^2}\ ,
\end{equation}
for any $\tilde\mu \geq 0$ consistent with
condition~(\ref{tocheck}). We now choose $\tilde{\mu} = 3
\sqrt{N_d}$, see Eq.~\eqref{l1}, getting
\begin{equation}\label{xx2}
    \Pr \left\{ \sum_\w F_\w \prod_{i=1}^{N_d} \Gamma_{w_i}^i
    \geq 3\sqrt{N_d} \prod_{i=1}^{N_d} \Omega_i \right\}
\leq
    {\rm e}^{-9N_d}\ .
\end{equation}
With this assignment, and using ~(\ref{nu*}) and~(\ref{tm}),
condition~(\ref{tocheck}), yet to be fulfilled, becomes
\begin{equation}\label{y}
    3\sqrt{N_d} \prod_{i=1}^{N_d} \frac{|\Gamma_{w_i}^i|}{\Omega_i}
    \leq 1\ ,
\end{equation}
which now holds because of the initial premise~(\ref{l12}).

Bound~(\ref{xx2}) applies to each of the sequences $(\a_1 ,\x_1),
\ldots (\a_{N_d} ,\x_{N_d})$ satisfying~(\ref{const4}), and there
are at most $4^{5 N_d}$ of them. Hence, the probability that the
random function $f$ does not satisfy the bound
\begin{equation}
    \sum_\w F_\w \prod_{i=1}^{N_d} \Gamma_{w_i}^i
    \geq 3\sqrt{N_d} \prod_{i=1}^{N_d} \Omega_i ,
\end{equation}
for at least one of such sequences, is at most $4^{5 N_d} {\rm
e}^{-9N_d}$, which is smaller than $1/2$ for any value of $N_d$. A
similar argument proves that the probability that the random
function $f$ does not satisfy the bound
\begin{equation}
    \sum_\w F_\w \prod_{i=1}^{N_d} \Gamma_{w_i}^i
    \leq -3\sqrt{N_d} \prod_{i=1}^{N_d} \Omega_i ,
\end{equation}
for at least one sequence satisfying~(\ref{const4}) is also
smaller than 1/2. The lemma now easily follows from these two
results.
\end{proof}

\section{Final remarks}
\label{Sec5}

The main goal of our work was to prove full randomness
amplification. In these appendices, we have shown how
our protocol, based on quantum non-local correlations, achieves
this task. Unfortunately, we are not able to provide an explicit
description of the function $f: \{0,1\}^{N_d} \to \{0,1\}$ which
maps the outcomes of the black boxes to the final random bit $k$;
we merely show its existence. Such function may be obtained
through an algorithm that searches over the set of all functions
until it finds one satisfying~(\ref{l1}). The problem with this
method is that the set of all functions has size $2^{N_d}$, which
makes the search computationally costly. However, this problem can
be fixed by noticing that the random choice of $f$ in the proof of
Lemma~\ref{L1} can be restricted to a four-universal family of
functions, with size polynomial in $N_d$. This observation will be
developed in future work.

A more direct approach could consist of studying how the
randomness in the measurement outcomes for correlations maximally
violating the Mermin inequality increases with the number of
parties. We solved linear optimization problems similar to those
used in Theorem~\ref{Theorem0} which showed that for 7 parties
Eve's predictability is $2/3$ for a function of 5 bits defined by
$f(00000)=0$, $f(01111)=0$, $f(00111)=0$ and $f(\x)=1$ otherwise.
Note that this value is lower than the earlier $3/4$ and also that
the function is different from the majority-vote. We were however
unable to generalize these results for an arbitrary number of
parties, which forced us to adopt a less direct approach. Note in
fact that our protocol can be interpreted as a huge multipartite
Bell test from which a random bit is extracted by classical
processing of some of the measurement outcomes.

We conclude by stressing again that the reason why randomness
amplification becomes possible using non-locality is because the
randomness certification is achieved by a Bell inequality
violation. There already exist several protocols, both in
classical and quantum information theory, in which imperfect
randomness is processed to generate perfect (or arbitrarily close
to perfect) randomness. However, all these protocols, e.g.
two-universal hashing or randomness extractors, always require
additional good-quality randomness to perform such distillation.
On the contrary, if the initial imperfect randomness has been
certified by a Bell inequality violation, the distillation
procedure can be done with a deterministic hash function
(see~\cite{PA} or Lemma~\ref{lem1} above). This property makes
Bell-certified randomness fundamentally different from any other
form of randomness, and is the key for the success of our
protocol.


\end{widetext}

\end{document}